\documentclass[11pt]{article}
\usepackage{authblk}
\usepackage{amssymb,latexsym,amsmath}
\usepackage{amsthm}

\theoremstyle{definition}
\newtheorem{theorem}{Theorem}
\newtheorem{corollary}[theorem]{Corollary}
\newtheorem{proposition}[theorem]{Proposition}
\newtheorem{lemma}[theorem]{Lemma}
\newtheorem{definition}[theorem]{Definition}
\newtheorem{example}[theorem]{Example}
\newtheorem{notation}[theorem]{Notation}
\newtheorem{remark}[theorem]{Remark}

\usepackage{multirow}
\usepackage{graphicx}
\usepackage{tikz}
\usetikzlibrary{matrix,decorations.pathreplacing}

\usepackage{lipsum}

\newcommand\blfootnote[1]{%
  \begingroup
  \renewcommand\thefootnote{}\footnote{#1}%
  \addtocounter{footnote}{-1}%
  \endgroup
}

\newcommand{\numberset}{\mathbb}

\newcommand{\N}{\numberset{N}}

\newcommand{\F}{\numberset{F}}

\newcommand{\mC}{\mathcal{C}}

\newcommand{\mG}{\mathcal{G}}

\usepackage{hyperref}
\newcommand{\mA}{\mathcal{A}}

\usepackage[margin=3cm]{geometry}

\DeclareMathOperator{\rdef}{Rdef}
\DeclareMathOperator{\mat}{Mat}

\newcommand{\mD}{\mathcal{D}}

\newcommand{\rk}{\mathrm{rk}}

\newcommand{\defect}{\mathrm{Rdef}}

\begin{document}

\title{Rank distribution of Delsarte codes\blfootnote{Email addresses: jdelacruz@uninorte.edu.co, elisa.gorla@unine.ch, hlopez@math.cinvestav.mx, alberto.ravagnani@unine.ch.}}

\author[1]{Javier de la Cruz}
\author[2]{Elisa Gorla}
\author[3]{Hiram H. L\'opez}
\author[2,*]{Alberto Ravagnani}

\affil[1]{Universidad del Norte, Colombia}
\affil[2]{Universit\'{e} de Neuch\^{a}tel, Switzerland}
\affil[3]{Centro de Investigaci\'{o}n y de Estudios Avanzados del IPN,
M\'{e}xico}

\footnotetext[1]{Part of the work was done while J. de la Cruz was visiting the University of Zurich. The first author thanks Joachim Rosenthal for the invitation.}
\footnotetext[2]{E. Gorla and A. Ravagnani were partially supported by the Swiss National Science Foundation through grant no. 200021\_150207 and by the ESF COST Action IC1104.}
\footnotetext[3]{H. L\'opez was partially supported by CONACyT and by the Swiss Confederation through the Swiss Government Excellence Scholarship no. 2014.0432.}

\date{}

\maketitle

\begin{abstract} 
In analogy with the Singleton defect for classical codes, we propose a definition of rank defect for Delsarte rank-metric codes. We characterize codes whose rank defect and dual rank defect are both zero, and prove that the rank distribution of such codes is determined by their parameters. This extends a result by Delsarte on the rank distribution of MRD codes. In the general case of codes of positive defect, we show that the rank distribution is determined by the parameters of the code, together the number of codewords of small rank. Moreover, we prove that if the rank defect of a code and its dual are both one, and the dimension satisfies a divisibility condition, then the number of minimum-rank codewords and dual minimum-rank codewords is the same. 
Finally, we discuss how our results specialize to Gabidulin codes.
\end{abstract}

\section*{Introduction}\label{intr0}

Rank-metric codes were first introduced in coding theory by Delsarte in~\cite{del1}. 
They are sets of matrices of fixed size, endowed with the rank distance. 
Rank-metric codes are of interest within network coding, 
public-key cryptography, and distributed storage, where they stimulated a series 
of works aimed at better understanding their properties.
In this paper, we study the rank distribution of rank-metric codes. 
We always assume that the codes are linear, and often refer to them as Delsarte codes.

The study of the weight distribution of a code is a topic of current interest in coding theory, where several authors 
have studied the case of linear codes endowed with the Hamming distance. In particular, it is a classical result that 
the weight distribution of an MDS code is determined by its parameters. The so-called MRD codes are the analogue 
of MDS codes in the context of Delsarte codes. They were introduced in~\cite{del1} by Delsarte, who also showed 
that the weight distribution of an MRD code is determined by its parameters. However, MRD codes only exist if the 
size of the matrix divides the dimension of the code. More precisely, for $n\times m$ matrices with $n\leq m$, 
$m$ must divide the dimension of the code.

In this paper, we study the rank distribution of Delsarte codes. In Section~\ref{def0} we define Quasi-MRD  (or QMRD) 
codes as codes which have the largest possible minimum distance for their parameters, but are not MRD. 
We regard them as the best alternative to MRD codes, for dimensions for which MRD codes do not exist. 
While the dual of an MRD code is MRD, the dual of a QMRD code is not necessarily QMRD. 
When both $\mC$ and its dual are QMRD, we say that $\mC$ is dually QMRD. 
In Proposition~\ref{dualqmrd} we provide a characterization of dually QMRD codes in terms of the number of codewords of minimum weight. 
Moreover, we show that QMRD codes exist for all choices of the parameters, and give examples of codes which 
are QMRD, but not dually QMRD. In analogy with the Singleton defect for classical codes, 
we propose a definition of rank defect for Delsarte codes. According to our definition, a code has rank defect zero 
if and only if it is either MRD or QMRD. 

Using the MacWilliams identities, in Section~\ref{sec2} we derive formulas 
that relate the numbers $A_i(\mC)$ of codewords of $\mC$ of weight $i$ for all values of $i$, showing that a few 
of the $A_i(\mC)$'s determine the others (see Theorem~\ref{04-02-15}). In analogy with the classical case of MDS codes, 
our result implies that the rank distribution of a code which is MRD or dually QMRD is completely determined by its parameters.
Notice that this is not the case in general for codes which are QMRD but not dually QMRD (so for codes of rank defect zero), 
as we show in Example~\ref{counterexdef0}. 

In Section~\ref{sec3} we analyze a family of codes such that both the code and its dual have rank defect one. 
In Theorem~\ref{bij} we show that they have the property that the code and its dual have the same number of minimum 
rank codewords. 

Finally, in Section~\ref{secgab} we consider Gabidulin codes and discuss how our
results specialize to this case.

\section{Preliminaries} \label{prel}

In this section we briefly recall the main definitions and results of the theory of
Delsarte rank-metric codes.

\begin{notation}
Throughout the paper, $q$ denotes a fixed prime power, and $\F_q$ the finite field with
$q$ elements. We also work with positive integers $1 \le n \le m$, and denote
by $\mbox{Mat}$ the vector space of $n \times m$ matrices with entries in $\F_q$. 
Given a positive integer $a$, we denote by $[a]$ the set $\{1,...,a\}$ and by $I_a$ the identity matrix of size $a$. 
The trace of a square matrix $M$ is denoted and defined by $\mbox{Tr}(M)=\sum_{i=1}^n M_{ii}$.
The rank of a matrix $M$ is denoted by $\mbox{rk}(M)$. 
All dimensions are computed over $\F_q$, unless otherwise specified.
\end{notation}

\begin{definition}
A (\textbf{Delsarte rank-metric}) \textbf{code} is an $\F_q$-linear subspace $\mC\subseteq\mbox{Mat}$. 
The \textbf{minimum distance} of a code $\mC \neq \{ 0\}$ is 
$d(\mC):=\min \{ \rk(M): M \in \mC, \ M\neq 0 \}$.  
The \textbf{rank distribution} of $\mC$ is the collection ${(A_i(\mC))}_{i \in \N}$, 
where $A_i(\mC):=|\{ M\in \mC: \rk(M)=i\}|$ for $i \in \N$.
The \textbf{dual} of a Delsarte code $\mC \subseteq \mbox{Mat}$ is the code 
$$\mC^\perp:=\{ M \in \mbox{Mat} : \mbox{Tr}(MN^t)=0 \mbox{
for all } N \in \mC\} \subseteq \mbox{Mat}.$$
A code $\mC$ is \textbf{trivial} if $\mC=\{ 0 \}$ or $\mC=\mbox{Mat}$.
\end{definition}

Throughout the paper, $\mC$ denotes a non-trivial Delsarte code with 
minimum distance $d$ and dimension $t$. We let $d^\perp$ be the minimum distance of its dual $\mC^\perp$.

\begin{remark} \label{props}
 The trace-product of matrices $(M,N) \mapsto \mbox{Tr}(MN^t)$ is a symmetric
 and non-degenerate bilinear form. In particular, the dual of a Delsarte code
$\mC$  is a Delsarte code of dimension $\dim(\mC^\perp)=mn-t$.
Moreover, given Delsarte codes $\mC,\mD \subseteq \mbox{Mat}$ we have
$(\mC + \mD)^\perp=\mC^\perp \cap \mD^\perp$ and $(\mC \cap
\mD)^\perp=\mC^\perp + \mD^\perp$. Finally, $(\mC^\perp)^\perp=\mC$ for any
Delsarte code $\mC$.
\end{remark}

The following result by Delsarte is well known. 
It is the analogue of the Singleton bound in the context of rank-metric codes.

\begin{theorem}[\cite{del1}, Theorem 5.4] \label{sbound}
Let $\mC$ be a Delsarte code with minimum distance $d$ and dimension $t$.
We have $t\le m(n-d+1)$.
\end{theorem}

\begin{definition}\label{defmrd}
A Delsarte code $\mC$ is \textbf{MRD} if $t= m(n-d+1)$.
\end{definition}

\begin{theorem}[\cite{del1}, Theorem 5.6] \label{dualmrd}
A Delsarte code $\mC$ is MRD if and only if the dual code $\mC^\perp$ is MRD.
\end{theorem}

We briefly recall the definition and the main algebraic properties of $q$-ary Gaussian
coefficients, for the convenience of the reader. A standard reference is  \cite{andr}.

\begin{definition} Let $q$ be a prime power, and let $a$ and $b$ be integers. 
  The $q$-ary Gaussian coefficient of $a$ and $b$ is defined by
  $${{a \brack b}} :=\left\{ \begin{array}{cl} 0 & \mbox{ if $a < 0, \; b <
 0$, or $b >a$,} \\ 1 & \mbox{ if $b=0$ and $a \geq 0$,} \\              
 \frac{(q^a-1)(q^{a-1}-1) \cdots (q^{a-b+1}-1)}{(q^b-1)(q^{b-1}-1) \cdots (q-1)} &
 \mbox{ otherwise.} \end{array} \right.$$
  \end{definition}

\begin{lemma} \label{progauss}
 Let $a,b,r$ be integers. The following hold:
  \begin{enumerate}
  \item ${a \brack 0} = {a \brack a}=1$ for $a \ge 0$, \label{bin1}
  \item ${a \brack b} = {a \brack a-b}$, \label{bin2}
  \item ${a \brack b} {b \brack r}={a \brack r} {a-r \brack a-b}$, \label{bin3}
  \item ${a \brack b}= q^b {a-1 \brack b}+ {a-1\brack b-1}$ for $a,b \ge 1$,
\label{bin4}
  \item $\sum_{i=0}^{a} (-1)^iq^{{i\choose 2}}{a \brack i}=0$ for $a \ge 1$.
\label{bin5}
 \end{enumerate}
\end{lemma}

In \cite{del1} Delsarte proved that the rank distribution of a code and of its dual satisfy MacWilliams identities. 
In particular, they determine each other. In this paper we use the following form of the MacWilliams identities 
for Delsarte codes, which was given in \cite[Corollary 1 and Proposition 3]{Gaduleau}. 
An elementary combinatorial proof can be found in \cite{Alberto}.

\begin{theorem}\label{MacWilliams identities} 
Let $\mC$ be a  Delsarte code. 
For any integer $0 \leq r \leq
n$ we have $$\sum_{i=0}^{n-r}A_i(\mC){n-i \brack
r}=\frac{|\mC|}{q^{mr}}\sum_{j=0}^{r}A_j(\mC^\perp){n-j
\brack  r-j }.$$
In particular, if $\mC$ is non-trivial then
$$\sum\limits_{i=d}^{n-r} 
{n-i\brack r}A_i(\mC) = \left( \frac{|\mC|}{q^{mr}}-1 \right) {n\brack r}$$ for $r=0, \ldots, d^\perp-1$,
where $d$ and $d^\perp$ denote  the minimum
distance of $\mC$ and $\mC^\perp$, respectively. 
\end{theorem}

\section{Rank defect, MRD, and quasi-MRD codes} \label{def0}

In this paper we study the rank distribution of codes which are MRD, or close to being MRD. 
Being MRD is the analogue of being MDS for codes endowed with the Hamming distance. 
A Hamming code is MDS if its minimum distance meets the Singleton bound. 
The natural analogue for rank-metric codes of the Singleton bound was given in Theorem~\ref{sbound}. In terms of minimum distance,  
Theorem~\ref{sbound} may be stated as follows.

\begin{corollary}\label{rdef0}
Let $\mC\subset\mat$ be a Delsarte code with minimum distance $d$ and dimension $t$.
Then $$d\leq n-\left\lceil\frac{t}{m}\right\rceil+1.$$ 
\end{corollary}

This motivates the definition of rank defect for Delsarte codes.

\begin{definition}\label{rkdef}
The {\bf rank defect} of $\mC$ is
$$\rdef(\mC)=n-\left\lceil\frac{t}{m}\right\rceil-d+1.$$
\end{definition}

For codes endowed with the Hamming distance, many authors regard the Singleton defect 
as a measure of how far the code is from being MDS (see e.g.
\cite{Faldum-Willems}). 
The situation is slightly more complicated in the case of rank-metric codes.
In fact, if $\mC$ is MRD then $\rdef(\mC)=0$. 
However, $\rdef(\mC)$ may be zero also for codes $\mC$ which are not MRD. 
%as we show in Example~\ref{trivialqmrd}.
This has a simple explanation: A code $\mC$ has $\rdef(\mC)=0$ 
if and only if its minimum distance has the largest possible value, for the given $m,n$, and $t$.
However, if $m\nmid t$, then $\mC$ is not MRD. 
This observation motivates the definition of Quasi-MRD code.

\begin{definition}
A code $\mC$ of dimension $t$ is {\bf Quasi-MRD}, or {\bf QMRD}, if $m\nmid t$ and 
$\mC$ attains the bound of Corollary~\ref{rdef0}. 
Equivalently, $\mC$ is QMRD if and only if $\rdef(\mC)=0$ and $\mC$ is not MRD.
\end{definition}

Existence of MRD codes was established in~\cite{del1} and~\cite{Gabidulin} 
for all $1\leq n\leq m$ and $1\leq d\leq n$. 
Recently, constructions of MRD codes which are not equivalent to
the previous ones appeared in~\cite{sheekey} and ~\cite{Willems-DelaCruz-K-W}.
We complete the picture by showing that QMRD codes exist for all choices of the
parameters.

\begin{example}[Existence of QMRD codes]\label{trivialqmrd}
For any $1\leq n\leq m$ and $1\leq t<nm$ such that $m\nmid t$, we can
construct a QMRD Delsarte code of dimension $t$ as follows. Let
$\mathcal{D}$ be an MRD code of dimension
$\dim(\mathcal{D})=m\left\lceil\frac{t}{m}\right\rceil$. $\mathcal{D}$ has
minimum distance
$$d=n-\frac{\dim\mathcal{D}}{m}+1=n-\left\lceil\frac{t}{m}\right\rceil+1.$$
Let $\mC\subseteq \mathcal{D}$ be a subspace of dimension $t$ containing a
codeword of $\mathcal{D}$ of minimum weight. Then $\mC$ has minimum distance
$d$, hence it is QMRD with the chosen parameters.
\end{example}

While the dual of an MRD code is MRD, the dual of a QMRD code is not necessarily QMRD,
as the following example shows.

\begin{example} \label{exdualqmrd}
 Let $q=2$, $n=m=3$. Let $\mC$ be the code generated over $\F_2$ by the four matrices
$$ \begin{bmatrix}
    0 & 1 & 0 \\ 1 & 0 & 0 \\ 0 & 0 & 0
   \end{bmatrix}, \ \ \ \ 
\begin{bmatrix}
    1 & 1 & 0 \\ 0 & 1 & 0 \\ 0 & 0 & 0
   \end{bmatrix}, \ \ \ \ 
\begin{bmatrix}
    0 & 0 & 1 \\ 0 & 0 & 0 \\ 1 & 0 & 0
   \end{bmatrix}, \ \ \ \ 
\begin{bmatrix}
    0 & 0 & 0 \\ 0 & 0 & 1 \\ 0 & 1 & 0
   \end{bmatrix}.
$$
$\mC$ has dimension 4 and minimum distance 2, therefore it is QMRD. On the other
hand,
 $$\begin{bmatrix} 
    0 & 0 & 0 \\ 0 & 0 & 0 \\ 0 & 0 & 1
   \end{bmatrix} \in \mC^\perp,
$$ so the minimum distance of $\mC^\perp$ is 1. In particular, $\mC^\perp$
is not QMRD.
\end{example}

The existence of such examples motivates the definition of dually MRD code.

\begin{definition}
A code $\mC$ is {\bf dually QMRD} if both $\mC$ and $\mC^\perp$ are QMRD.
\end{definition}

One can also find examples over any ground field and for matrices of size $n\times m$, if $m\geq 2$.

\begin{example}\label{duallyQMRD}
Let $2\leq n\leq m$, $0<\rho<m$, and let $\mC$ be a code of dimension $\rho$ and 
minimum distance $d<n$. $\mC$ is not QMRD, since 
$$\rdef(\mC)=n-\left\lceil \frac{\rho}{m}\right\rceil-d+1\neq 0.$$
Its dual $\mC^\perp$ has dimension $mn-\rho$ and minimum distance
$$d^\perp\leq n-\left\lceil n-\frac{\rho}{m}\right\rceil+1=1,$$ by Corollary~\ref{rdef0}.
Hence $d^\perp=1$ and $\mC^\perp$ is QMRD. 
\end{example}
 
We can characterize dually QMRD codes in terms of their number of minimum-rank codewords.

\begin{proposition} \label{dualqmrd}
Let $\mC$ be QMRD of dimension $t$, and let $0<\rho<m$ be the reminder obtained dividing
$t$ by $m$. Then $\mC^\perp$ is QMRD if and only if $A_d(\mC)= {n\brack d} (q^\rho-1)$.
\end{proposition}

\begin{proof}
 Write $t=\alpha m + \rho$ with $\alpha \in \N$. Since $\mC$ is QMRD, we
have $d=n-\alpha$. Moreover, since $\dim(\mC^\perp)=nm-t$, the code
$\mC^\perp$ is QMRD if and only if it has minimum distance $d^\perp=\alpha+1$. 
Theorem \ref{MacWilliams identities} with $r=\alpha$ gives
\begin{equation} \label{eqq}
 {n\brack \alpha}+A_d(\mC)=q^\rho \left( {n\brack \alpha} + \sum_{j=1}^\alpha A_j(\mC^\perp) 
{n-j\brack \alpha-j}\right).
\end{equation}
Since $\alpha \le n$, we have ${n-j\brack \alpha-j}>0$ for $j \in \{1,...,\alpha\}$. Therefore
$\mC^\perp$ is QMRD if and only if $\sum_{j=1}^\alpha A_j(\mC^\perp) 
{n-j\brack \alpha-j}=0$. 
Hence $\mC^\perp$ is QMRD if and only if
$$A_d(\mC)=(q^\rho-1) {n\brack \alpha}.$$
By Lemma~\ref{progauss} we have  ${n\brack \alpha}= {n\brack d}$, and the result follows.
\end{proof}

\begin{remark}
Following the notation of Proposition \ref{dualqmrd}, the code $\mC$ of Example \ref{exdualqmrd}
has $\rho=1$. One can check that $A_d(\mC)=9 \neq  {3\brack 2} (2^1-1)=7$, in 
accordance with the fact that
$\mC^\perp$ is not QMRD.
\end{remark}

The following is a simple consequence of Theorem~\ref{sbound}.

\begin{corollary} \label{d+2} 
Let $\mC$ be a Delsarte code of dimension $t$, with minimum distance $d$ and dual minimum distance $d^\perp$. 
The following hold:
\begin{enumerate}
\item If $m\mid t$, then either $d+d^\perp=n+2$ or $d+d^\perp \le n$.
\item If $m\nmid t$, then $d+d^\perp \le n+1$.
\end{enumerate}
\end{corollary}

\begin{proof}
Theorem \ref{sbound} applied to $\mC$ and $\mC^\perp$ gives 
\begin{equation} \label{1}
t \le
m(n-d+1) \ \ \ \ \ \ \ \mbox{and} \ \ \ \ \ \ \ nm-t \le m(n-d^\perp+1).
\end{equation}
Summing the two inequalities in (\ref{1}) we obtain $n \le 2n-d-d^\perp+2$, i.e. 
$d+d^\perp \le n+2$.
In addition, $d+d^\perp = n+2$ if and only if both inequalities are equalities, 
which implies that $\mC$ is MRD, \
hence $m\mid t$. This proves part 2. If $t=km$, then the 
inequalities (\ref{1}) become
$k \le n-d+1$ and $n-k \le n-d^\perp+1$. Hence if $\mC$ is not MRD we have
$k \le n-d$ and $n-k \le n-d^\perp$. It follows that $d+d^\perp \le n$.
\end{proof}

It is now easy to characterize the codes $\mC$ such that both $\mC$ and $\mC^\perp$ have
rank defect zero as those which are MRD or dually QMRD.
In Corollary \ref{compl_det} we will show that the
rank distribution of such codes is determined by $n$, $m$ and $d$.

\begin{proposition} \label{carattdually} 
Let $\mC$ be a Delsarte code with minimum distance $d$ and dual minimum distance $d^\perp$. 
The following hold:
\begin{enumerate}
\item $\mC$ is MRD iff $\mC^\perp$ is MRD iff $d+d^\perp=n+2$.
\item $\mC$ is dually QMRD iff $d+d^\perp=n+1$.
\end{enumerate}
\end{proposition}

\begin{proof}
1. In the proof of Corollary~\ref{d+2} we showed that $d+d^\perp=n+2$ implies $\mC$ MRD hence, 
by symmetry, $\mC^\perp$ MRD. Conversely, if one of $\mC$ and $\mC^\perp$ is MRD, then the 
dual is also MRD 
by Theorem~\ref{dualmrd}. Therefore 
$$d+d^\perp=n-\frac{t}{m}+1+n-\frac{nm-t}{m}+1=n+2.$$
2. If $d+d^\perp=n+1$, then $m\nmid t$. The bound of Corollary~\ref{rdef0} 
yields
$$d\leq n-\left\lceil\frac{t}{m}\right\rceil+1 \ \ \ \ \mbox{ and } \ \
\ \ d^\perp\leq \left\lfloor\frac{t}{m}\right\rfloor+1.$$ Therefore both
inequalities are equalities and $\mC$ and $\mC^\perp$ are QMRD.
Conversely, if $\mC$ and $\mC^\perp$ are QMRD, then $m\nmid t$ and
$$d+d^\perp=n-\left\lceil\frac{t}{m}\right\rceil+1+\left\lfloor\frac{
t}{m}\right\rfloor+1=n+1.$$
\end{proof}

Generalized weights for Delsarte codes were introduced in \cite{a2}, refining 
previous definitions for Gabidulin codes.
We conclude this section with some observations on the connection between 
rank defect and generalized weights. 

\begin{definition}
An \textbf{optimal anticode} $\mA \subseteq \mbox{Mat}$ is a Delsarte code such that
$\dim(\mA)=m \cdot \mbox{maxrk}(\mA)$, where
$\mbox{maxrk}(\mA):= \max \{ \mbox{rk}(M) : M \in \mA\}$.

Given a Delsarte code $\mC$ of dimension $t$ and an integer $1 \le r \le t$, the \textbf{$r$-th generalized weight}
of $\mC$ is
$$a_r(\mC):= \frac{1}{m} \min \{ \dim(\mA) : \mA \subseteq \mbox{Mat} \mbox{ is an anticode with 
} \dim(\mC \cap \mA) \ge r\}.$$
\end{definition}

The next result relates the generalized weights of a code to the 
minimum distance and the rank defect of the dual code.

\begin{proposition}
Let $\mC$ be a Delsarte code with minimum distance $d$ and dimension $t$.
If $t < m$ then $d^\perp=1$ and $\rdef(\mC^\perp)=0$. If $t \ge m$,
write $t=\alpha m + \rho$ with $0 \le \rho < m$.
The minimum distance of the dual code $\mC^\perp$ is
$$d^\perp= \left\{ \begin{array}{ll}
                    \alpha+1 & \mbox{ if } n+1-a_{t+1-um}(\mC)=\alpha, \\ 
\min \{ 1 \le r \le \alpha : n+1-a_{t+1-rm}(\mC)>r\} & \mbox{ otherwise,} \end{array} \right.\ $$
and the rank defect of $\mC^\perp$ is
$$\rdef({\mC^\perp})=\alpha+1-d^\perp.$$
\end{proposition}

\begin{proof}
If $t < m$, then $d^\perp=1$ by  Corollary \ref{rdef0}. Assume therefore that
$t \ge m$. By Theorem 66 of
\cite{a2} we have
$$d^\perp=a_1(\mC^\perp) < a_{1+m}(\mC^\perp) < \cdots < a_{1+(n-\alpha-1)m}(\mC^\perp) \le n.$$
Define
$$W_1(\mathcal{C}^\perp):=\left\{a_1(\mathcal{C}^\perp),a_{1+m}(\mathcal{C}
^\perp),\ldots,a_{1+(n-\alpha-1)m}(\mathcal{C}^\perp)\right\}$$
and
$$\overline{W}_{1+t}(\mathcal{C}):=\left\{n+1-a_{t+1-m}(\mathcal{C}),n+1-a_{
t+1-2m}(\mathcal{C}),\ldots,n+1-a_{t+1-\alpha m}(\mathcal{C})\right\}.$$
By \cite[Corollary 78]{a2} we have that 
$W_1(\mathcal{C}^\perp)=[n]\setminus \overline{W}_{1+t}(\mathcal{C})$. Hence 
the result on the minimum distance follows from the fact that
$$n+1-a_{t+1-m}(\mathcal{C}) < n+1-a_{
t+1-2m}(\mathcal{C}) < \ldots < n+1-a_{t+1-\alpha m}(\mathcal{C})$$
by Theorem 66 of \cite{a2}. The formula for the rank defect now follows from the definition.
\end{proof}

\section{Rank distribution of Delsarte codes} \label{sec2}

In this section we prove that the rank distribution of a code $\mC$ is determined by 
its parameters, together with the number of codewords of small weight: 
$A_d(\mC),\ldots, A_{n-d^\bot}(\mC)$. In particular, we show that the rank distribution of
MRD and dually QMRD codes only depends on their parameters. We start with a
series of preliminary definitions and results.

\begin{lemma}\label{matrix-inverse}
Let $a \ge 1$ be an integer and let $M,N$ be the real $a \times a$ matrices
defined by $M_{ij}:= {j-1 \brack i-1}$ and ${N}_{ij}:=(-1)^{j-i}q^{j-i \choose
2}{j-1\brack i-1}$  {for $i,j \in \left[a\right]$}.
Then $MN=NM=I_{a}.$
\end{lemma}

\begin{proof}
Since $M$ and $N$ are square matrices, it suffices to prove that $MN=I_a$.
One can easily check that
$${(MN)}_{ij} =\sum_{r=0}^{a-1}(-1)^{j-r}q^{{j-r\choose2}}{j\brack r}{r\brack
i}.$$
By Lemma \ref{progauss}.\ref{bin3} we have 
 \begin{equation*}
   {(MN)}_{ij}= \sum_{r=0}^{a-1}(-1)^{j-r}q^{{j-r\choose2}}{j\brack r}{r\brack
i} = {j \brack i}\sum_{r=0}^j(-1)^{j-r}q^{{j-r\choose2}}{j-i\brack j-r}.
  \end{equation*}
Hence  ${(MN)}_{ii}=1$ for $i=1,...,{a}$. If $i>j$, then ${(MN)}_{ij}=0$. 
If $i<j$, then
\begin{eqnarray*}
 {(MN)}_{ij} &=& {j \brack i}\sum_{r=0}^j(-1)^{j-r}q^{{j-r\choose2}}{j-i\brack
j-r} \\
&=& {j \brack i}\sum_{r=0}^j(-1)^{r}q^{{r\choose2}}{j-i\brack r} \\
&=& {j \brack i}\sum_{r=0}^{j-i}(-1)^{r}q^{{r\choose2}}{j-i\brack r} \\
&=& 0,
\end{eqnarray*}
where the last equality follows from Lemma \ref{progauss}.\ref{bin5}.
\end{proof}

\begin{lemma}\label{sume}
Let $a \in \N_{\ge 1}$. For $j = 0,..., a$ we have
$$\sum_{i=0}^j (-1)^i q ^{{i \choose 2}}{a \brack i}=(-1)^j q^{ {j+1\choose
2}}{a-1 \brack j}.$$
\end{lemma}

\begin{proof}
By induction on $j$. If $j = 0$ then the result is immediate.
Assume $j>0$. By induction hypothesis and Lemma~\ref{progauss}.\ref{bin4}
we have
\begin{eqnarray*}
 \sum_{i=0}^j (-1)^i q ^{{i \choose 2}}{a \brack i} &=& (-1)^{j-1} q^{{ j
\choose 2}} {a-1 \brack j-1}+ (-1)^j q^{{ j \choose 2}} {a \brack j} \\
&=& (-1)^{j} q^{{ j+1 \choose 2}} q^{-j} \left( {a \brack j}-{a-1 \brack j-1} 
\right) \\
&=& (-1)^{j} q^{{ j+1 \choose 2}}{a-1 \brack j},
\end{eqnarray*}
as claimed.
\end{proof}

Now we state our main result.

\begin{theorem}\label{04-02-15}
Let $\mC$ be a $t$-dimensional code, with minimum distance $d$ and dual minimum distance $d^\perp$.
Let $\delta=1$ if $d+d^\perp=n+2$, and $\delta=0$ otherwise.
For all $1 \le r \le d^\perp$ we have
\begin{eqnarray*}
A_{n-d^\bot+r}(\mC) &=& (-1)^{r}q^{r \choose 2}
\sum\limits_{j=d^\bot}^{n-d}{j\brack d^\bot-r}  {j-d^\bot+r-1\brack
r-1}A_{n-j}(\mC)\\
&& + {n\brack d^\bot-r}\sum\limits_{i=0}^{r-1-\delta} (-1)^iq^{i \choose 2
}{n-d^\bot+r\brack i}    \left( \frac{|\mC|}{q^{m(d^\perp-r+i)}}  -1 \right).
\end{eqnarray*}
In particular, $n$, $m$, $t$, $d$, $d^\perp$ and $A_d(\mC),\ldots, A_{n-d^\bot}(\mC)$ 
determine the rank
distribution of $\mC$.
\end{theorem}

\begin{proof} 
We only prove the theorem in the case $d+d^\perp\leq n+1$. 
The proof in the case $d+d^\perp=n+2$ is analogous and easier.
Let $A, B$ be the real matrices of size $d^\bot \times d^\bot$ and $d^\bot
\times \left(n-d-d^\bot+1\right)$, defined by $A_{r,j}={j-1\brack r-1}$ and
$B_{r,i}={i+d^\bot-1\brack r-1}$
for $r,j\in\left[d^\bot\right]$ and $i\in\left[n-d-d^\bot+1\right]$.
When $d+d^\perp = n+1$ we only have the matrix $A$.
Throughout the proof we write $A_i$ for $A_i(\mC)$.
The second part of  the statement of Theorem~\ref{MacWilliams identities} 
may be written in the form 
$$\left(A \left| B \right)\right.
\left(A_n,\ldots,A_d\right)^t
=\left(\left(|\mC|-1\right){n\brack
0},\ldots,\left(\frac{|\mC|}{q^{m(n-d)}}-1\right){n\brack
d^\perp-1}\right)^t.$$
Multiplying by $A^{-1}$ we get
$$(I_{d^\bot} |A^{-1}B)\left(A_n,\ldots,A_d\right)^t
= A^{-1}\left(\left(|\mC|-1\right){n\brack
0},\ldots,\left(\frac{|\mC|}{q^{m(n-d)}}-1\right){n\brack
d^\perp-1}\right)^t.
$$
By Lemma~\ref{matrix-inverse}, ${(A^{-1})}_{r,j}= (-1)^{j-r}q^{j-r \choose
2}{j-1\brack r-1}$ for $r,j \in \left[d^\bot\right]$. Hence
for $1\leq r<d^\bot+1 \leq j \leq n-d+1$ the entry of
the matrix $(I_{d^\bot} |A^{-1}B)$ in position $\left(r,j\right)$ is 
\begin{eqnarray*}
(I_{d^\bot} |A^{-1}B)_{r,j}&=&\sum\limits_{i=0}^{d^\bot-1} (-1)^{i-r+1}q^{i-r+1
\choose 2}{i\brack r-1}{j-1\brack i}\\
&=& \sum\limits_{i=0}^{n-d} (-1)^{i-r+1}q^{i-r+1 \choose 2}{i\brack
r-1}{j-1\brack i}\\
&&-\sum\limits_{i=d^\bot}^{n-d} (-1)^{i-r+1}q^{i-r+1 \choose 2}{i\brack
r-1}{j-1\brack i}\\
&=& \delta_{r-1,j-1}-\sum\limits_{i=d^\bot}^{n-d} (-1)^{i-r+1}q^{i-r+1 \choose
2}{i\brack r-1}{j-1\brack i}\\
&=& - \sum\limits_{i=d^\bot}^{n-d} (-1)^{i-r+1}q^{i-r+1 \choose 2}{i\brack
r-1}{j-1\brack i}\\
&=& -\sum\limits_{i=d^\bot}^{j-1} (-1)^{i-r+1}q^{i-r+1 \choose 2}{i\brack
r-1}{j-1\brack i}.
\end{eqnarray*}
As a consequence, for $r=1,\ldots,d^\bot$ we obtain
\begin{eqnarray*}
A_{n-r+1}+\sum_{j=d^{\perp}+1}^{n-d+1}\left(-\sum_{i=d^{\perp}}^{j-1}(-1)^{i-r+1
}
q^{i-r+1 \choose 2}{i\brack r-1}{j-1\brack i}\right)A_{n-j+1}&=&\\
=\sum\limits_{j=1}^{d^\bot} (-1)^{j-r}q^{j-r \choose 2}{j-1\brack r-1}{n\brack
j-1}   \left( \frac{|\mC|}{q^{m(j-1)}}-1 \right)&&
\end{eqnarray*} 
or, equivalently, 
\begin{eqnarray*}
A_{n-d^\bot+r}
&+&\sum\limits_{j=d^\bot}^{n-d}\left(-\sum\limits_{i=d^\bot}^{j}
(-1)^{i-d^\bot+r}q^{i-d^{\perp}+r \choose 2}{i\brack d^\bot-r}{j\brack
i}\right)A_{n-j}=\\
&=&\sum\limits_{j=0}^{d^\bot-1} (-1)^{j-d^\bot+r}q^{j-d^{\perp}+r \choose
2}{j\brack d^\bot-r}{n\brack j}   \left(  \frac{|\mC|}{q^{mj}} -1  \right).\\
\end{eqnarray*}
By  Lemma~\ref{progauss} we have
\begin{eqnarray*}
A_{n-d^\bot+r} &=& \sum\limits_{j=d^\bot}^{n-d}\sum\limits_{i=d^\bot}^{j}
(-1)^{i-d^\bot+r}q^{i-d^{\perp}+r \choose 2}{i\brack d^\bot-r}{j\brack
i}A_{n-j}\\
&& +\sum\limits_{j=0}^{d^\bot-1} (-1)^{j-d^\bot+r}q^{j-d^{\perp}+r \choose
2}{j\brack d^\bot-r}{n\brack j}   \left(  \frac{|\mC|}{q^{mj}}  -1 \right)\\
&=& \sum\limits_{j=d^\bot}^{n-d}\sum\limits_{i=d^\bot}^{j}
(-1)^{i-d^\bot+r}q^{i-d^{\perp}+r \choose 2}{i\brack d^\bot-r}{j\brack
i}A_{n-j}\\
&& +\sum\limits_{i=0}^{r-1} (-1)^iq^{i\choose 2}{d^\bot-r+i\brack
d^\bot-r}{n\brack d^\bot-r+i}     \left(  \frac{|\mC|}{q^{m(d^\bot-r+i)}} -1 \right)\\
&=& \sum\limits_{j=d^\bot}^{n-d}{j\brack
d^\bot-r}\left(\sum\limits_{i=d^\bot}^{j} (-1)^{i-d^\bot+r}q^{i-d^{\perp}+r
\choose 2}{j-d^\bot+r\brack j-i}\right)A_{n-j}\\
&& + {n\brack d^\bot-r}\sum\limits_{i=0}^{r-1} (-1)^iq^{i\choose
2}{n-d^\bot+r\brack i}   \left(  \frac{|\mC|}{q^{m(d^\bot-r+i)}} -1 \right).
\end{eqnarray*}
Using Lemma~\ref{progauss} and Lemma~\ref{sume}, the first term of the sum can be 
simplified as follows:
\begin{eqnarray*}
\sum\limits_{i=d^\bot}^{j} (-1)^{i-d^\bot+r}q^{i-d^{\perp}+r \choose
2}{j-d^\bot+r\brack j-i}
&=& \sum\limits_{i=r}^{j-d^\bot+r} (-1)^{i}q^{i \choose 2}{j-d^\bot+r\brack
j-d^\bot+r-i} \\
&=& \sum\limits_{i=r}^{j-d^\bot+r} (-1)^{i}q^{i \choose 2}{j-d^\bot+r\brack i}
\\ &=& -\sum\limits_{i=0}^{r-1} (-1)^{i}q^{i \choose 2}{j-d^\bot+r\brack i} \\
&=& (-1)^{r}q^{r \choose 2} {j-d^\bot+r-1\brack r-1}.
\end{eqnarray*}
The theorem follows by combining the equalities.
\end{proof}

%The result stated in Theorem \ref{04-02-15} is particularly simple in the case of dually QMRD codes. 
Theorem~\ref{04-02-15}, from which the next corollary easily follows, extends Theorem 5.6 of \cite{del1}
on the rank distribution of MRD codes.

\begin{corollary} \label{compl_det}
Assume that $\mC$ is MRD or dually QMRD. Then the rank distribution of $\mC$ is
given by
 $$A_{r}(\mC) = {n\brack r} \sum_{i=0}^{r-d}
  {(-1)}^{i}
  q^{\binom{i}{2}}
  {r \brack i}
 \left(\frac{|\mC|}{q^{m(n+i-r)}}-1\right)$$
 for $r=d,...,n$. In particular, it is completely determined by $n,m$, and $d$.
\end{corollary}

Finally, we show that Corollary~\ref{compl_det} does not hold for codes which are QMRD, but not dually QMRD.
In particular, for such codes the parameters $m,n$, and $d$ do not determine the weight distribution.

\begin{example}\label{counterexdef0}
Let $2\leq n\leq m$, and let $\mC$ be a code of dimension $1$ and minimum distance $d<n$. 
In Example~\ref{duallyQMRD} we showed that $\mC^\perp$ is QMRD, but not dually QMRD.
Theorem~\ref{MacWilliams identities} for $r=1$ yields the identity $${n \brack
1}+A_d(\mC){n-d \brack 1}=q^{1-m}\left( {n \brack  1 }+A_1(\mC^\perp){n-1 \brack 0}\right)$$ from which 
$$A_1(\mC^\perp)=\frac{q^{m+n-1}+q^{m+n-d}-q^{m+n-d-1}-q^m-q^n+1}{q-1}.$$
In particular, the weight distribution of $\mC^\perp$ depends on $d$ and not only on $m,n,d^\perp$.
\end{example}

\section{Codes with small rank defect} \label{sec3}

In Corollary \ref{compl_det}  we established a very special property of non-trivial 
Delsarte codes $\mC$ whose minimum distance and dual minimum distance satisfy $n+1 \le d+d^\perp \le n+2$. 
In this section we study codes $\mC$ such that $d+d^\perp=n$ and 
$m\mid t$, proving that the sets of minimum-rank codewords of 
$\mC$ and $\mC^\perp$ have the same cardinality. 

\begin{notation}
Given a subspace $U \subseteq \F_q^n$, we 
set $\mbox{Mat}(U):=\{ M \in \mbox{Mat} : \mbox{colsp}(M) \subseteq U\}$ and 
$\mC(U):=\mC
\cap \mbox{Mat}(U)$. The dual of a subspace $U \subseteq \F_q^n$ with respect to
the standard
inner product of $\F_q^n$ is denoted by $U^\perp$.
\end{notation}

We need the following technical result.

\begin{lemma}[\cite{Alberto}, Remark 22 and Lemma 24] \label{dimdu}
 Let $U \subseteq \F_q^n$ be an $\F_q$-subspace. The following hold:
\begin{enumerate}
 \item $\mbox{Mat}(U)$ is an $\F_q$-vector subspace of $\mbox{Mat}$ with
$\dim(\mbox{Mat}(U))=m\cdot \dim(U)$. \label{p1}
\item $\mbox{Mat}(U)^\perp=\mbox{Mat}(U^\perp)$. \label{p2}
\end{enumerate}
  \end{lemma}

For a given subspace $U \subseteq \F_q^n$, the dimension of $\mC(U)$
and the dimension of $\mC^\perp(U^\perp)$ relate as follows.

\begin{proposition} \label{formula}
Let $\mC$ be a Delsarte code of dimension $t$, let $U \subseteq
\F_q^n$ be a subspace. We have
$$\dim(\mC(U))=\dim(\mC^\perp(U^\perp))+t-m(n-\dim(U)).$$
\end{proposition}

\begin{proof}
By Remark \ref{props}, 
$\dim(\mC(U)^\perp)=mn-\dim(\mC(U))$. On the other hand,  by 
Remark \ref{props} and Lemma \ref{dimdu}.\ref{p2} we have 
$\mC(U)^\perp=\mC^\perp+\mbox{Mat}(U^\perp)$. By Lemma \ref{dimdu}.\ref{p1} we have
\begin{eqnarray*}
 mn-\dim(\mC(U)) &=& \dim(\mC(U)^\perp) \\ &=&
\dim(\mC^\perp) + \dim(\mbox{Mat}(U^\perp))-\dim(\mC^\perp
\cap \mbox{Mat}(U^\perp)) \\
&=& mn-t + m (n-\dim(U)) -\dim(\mC^\perp(U^\perp)).
\end{eqnarray*}
The proposition follows.
\end{proof}

\begin{theorem} \label{bij}
Let $\mC$ be a $t$-dimensional code, with minimum distance $d$ and dual minimum distance $d^\perp$.
Assume that $d+d^\perp=n$ and $m\mid t$. Then $\rdef(\mC)=\rdef(\mC^\perp)=1$ and 
$$A_d(\mC)=A_{d^\perp}(\mC^\perp).$$
\end{theorem}

\begin{proof}
Since
$\mC$
has minimum distance $d$, for all subspaces $U,U' \subseteq \F_q^n$ with
$\dim(U)=\dim(U')=d$ and $U \neq U'$ we have
$\mC(U) \cap \mC(U')=\{ 0\}$. Similarly, for all subspaces $U,U' \subseteq
\F_q^n$ with
$\dim(U)=\dim(U')=d^\perp=n-d$ and $U \neq U'$ we have
$\mC^\perp(U) \cap \mC^\perp(U')=\{ 0\}$. As a consequence, the number of
minimum-rank codewords
of $\mC$ and $\mC^\perp$ is, respectively,
$$A_d(\mC)=\sum_{\shortstack{$\scriptstyle U \subseteq \F_q^n$\\$\scriptstyle 
\dim_{\F_q}(U)=d$}} (|\mC(U)|-1), \ \ \ \ \ \ \ \ \ \ \ 
A_{d^\perp}(\mC^\perp)=\sum_{\shortstack{$\scriptstyle U \subseteq
\F_q^n$\\$\scriptstyle 
\dim_{\F_q}(U)=n-d$}} (|\mC^\perp(U)|-1).$$
Write $t=mk$ with $k \in \N$. Since $d+d^\perp=n$, then $\mC$ is not 
MRD by Corollary \ref{d+2}.
Theorem \ref{sbound} applied to $\mC$ and $\mC^\perp$ 
gives $k\le n-d$ and $n-k \le n-d^\perp$. Since $d+d^\perp=n$, then
we have $k=n-d=d^\perp$, and $\rdef(\mC)=\rdef(\mC^\perp)=1$. 
Hence for any subspace $U \subseteq \F_q^n$ with
$\dim(U)=d$  we have
$$t-m(n-\dim(U))=m(k-n+d)=0.$$
Therefore, $\dim(\mC(U))=\dim(\mC^\perp(U^\perp))$ by Proposition \ref{formula}.
It follows that
\begin{eqnarray*}
 A_d(\mC) &=& \sum_{\shortstack{$\scriptstyle U \subseteq
\F_q^n$\\$\scriptstyle 
\dim_{\F_q}(U)=d$}} (|\mC(U)|-1) \\
&=& \sum_{\shortstack{$\scriptstyle U \subseteq \F_q^n$\\$\scriptstyle 
\dim_{\F_q}(U)=d$}} (|\mC^\perp(U^\perp)|-1) \\
&=& \sum_{\shortstack{$\scriptstyle U \subseteq \F_q^n$\\$\scriptstyle 
\dim_{\F_q}(U)=n-d$}} (|\mC^\perp(U)|-1) \\
&=& A_{d^\perp}(\mC^\perp),
\end{eqnarray*}
as claimed.
\end{proof}

\begin{remark}
In Theorem~\ref{bij} we prove that, if $d+d^\perp=n$ and $m\mid t$, then $\rdef(\mC)=\rdef(\mC^\perp)=1.$
If instead $d+d^\perp=n$ and $m\nmid t$, then it is easy to show that $\rdef(\mC)+\rdef(\mC^\perp)=1.$
\end{remark}

\section{Gabidulin codes} \label{secgab}

In this section we discuss how the results from Section \ref{sec2} specialize 
to Gabidulin codes.

\begin{definition} 
A (\textbf{rank-metric Gabidulin}) \textbf{code} of length $n$ and dimension $0 \le k \le n$
is a $k$-dimensional $\F_{q^m}$-subspace $C \subseteq \F_{q^m}^n$. The
\textbf{rank} of a vector
$v=(v_1,...,v_n) \in \F_{q^m}^n$ is 
$$\rk(v):= \dim_{\F_q} \mbox{span}_{\F_q} \{ v_1,...,v_n\}.$$
The \textbf{minimum distance}
of a code $C \neq \{ 0\}$ is $$d(C):=\min \{ \rk(v): v \in C, \ v \neq 0
\}.$$
The \textbf{rank distribution} of a code $C$ is the collection ${(A_i(C))}_{i \in \N}$, 
where $$A_i(C):=|\{ v\in C: \rk(v)=i\}|.$$
The \textbf{dual} of a Gabidulin code $C \subseteq \F_{q^m}^n$ is the Gabidulin
code
$$C^\perp:=\{ v \in \F_{q^m}^k : \langle c,v \rangle =0 \mbox{ for all } c \in C\} 
\subseteq \F_{q^m}^n,$$
where $\langle \cdot, \cdot \rangle$ is the standard inner product of
$\F_{q^m}^n$. A code is \textbf{trivial} if $C=\{0\}$ or $C=\F_{q^m}^n$.
\end{definition}

\begin{notation}
Throughout this section, $C$ denotes a non-trivial Gabidulin code $C \subseteq \F_{q^m}^n$ 
with minimum distance $d$, dual minimum distance $d^\perp$, of dimension $k$ over $\F_{q^m}$. 
We assume that $n\leq m$.
\end{notation}

There is a natural way to associate a Delsarte code to a Gabidulin code.

\begin{definition} \label{assoc}
 Let $v=(v_1,...,v_n) \in \F_{q^m}^n$, and let $\mG=\{ \gamma_1,...,\gamma_m\}$
be a basis of $\F_{q^m}$ over $\F_q$. The matrix \textbf{associated} to $v$ with respect to
the basis $\mG$ is the $n \times m$ matrix $M_\mG(v)$ with entries in $\F_q$ such that
$$v_i= \sum_{j=1}^m M_\mG(v) \gamma_j.$$ The Delsarte code \textbf{associated} 
 to the Gabidulin code $C \subseteq \F_{q^m}^n$ with respect to the basis $\mG$ is
$$\mC_\mG(C):= \{ M_\mG(v): v\in C\}.$$
\end{definition}

We will need the following properties of associated Delsarte codes.

\begin{theorem} \label{serie}
Let $\mC$ be a Delsarte code associated to the Gabidulin code $C \subseteq \F_{q^m}^n$. Then:
\begin{enumerate}
\item $\dim(\mC)=mk$.
 \item $C$ has the same rank distribution as $\mC$. In particular, if $C$ is non-zero 
then they have the same minimum distance.
 \item $C^\perp$ has the same rank distribution as $\mC^\perp$.
\end{enumerate}
\end{theorem}

\begin{proof}
The first two parts of the statement easily follows from Definition \ref{assoc}. The third part 
is Theorem 18 of \cite{Alberto}.
\end{proof}

Combining  Theorem \ref{sbound} and Theorem \ref{serie} we obtain the following result.

\begin{proposition}\label{ssbound}
Let $C\subseteq \F_{q^m}^n$ be a Gabidulin code with minimum distance $d$ and dimension $k$. Then
$k \le n-d+1$. 
\end{proposition}

\begin{definition}
A Gabidulin code $C \subseteq \F_{q^m}^n$ with minimum distance $d$ and dimension $k$ is \textbf{MRD} if
$k=n-d+1$.
\end{definition}

It is well known that  a Gabidulin code $C$ is MRD if and only if $C^\perp$ is MRD.
We obtain the next result by combining Corollary \ref{d+2} and Theorem \ref{serie}.

\begin{proposition}
Let $C\subseteq \F_{q^m}^n$ be a Gabidulin code with minimum distance $d$ and dual minimum distance $d^\perp$.
One of the following holds:
\begin{enumerate}
 \item $C$ is MRD, and $d+d^\perp=n+2$.
\item $d+d^\perp \le n$.
\end{enumerate}
\end{proposition}

We notice in  particular that the case $d+d^\perp=n+1$ does not occur for Gabidulin codes. 
Hence Corollary  
\ref{compl_det} reads as follows for Gabidulin codes.

\begin{corollary}
Let $C\subseteq \F_{q^m}^n$ be a Gabidulin code with minimum distance $d$ and dimension $k$.
 If $C$ is MRD, then the rank distribution of $C$ satisfies 
$$A_{r}(C) = {n\brack r} \sum_{i=0}^{r-d}
 {(-1)}^{i}
 q^{\binom{i}{2}}
 {r \brack i}
\left(q^{m(k-n+i-r)}-1\right)$$
for $r=d,...,n$. In particular, it is completely determined by $n$, $m$,
and $d$.
\end{corollary}

Similarly, Theorem \ref{04-02-15} and Theorem \ref{bij} read as follows for Gabidulin codes.

\begin{corollary} \label{def_gabid}
Let $C\subseteq \F_{q^m}^n$ be a Gabidulin code with minimum distance $d$ and dimension $k$.
 If $d+d^\perp \le n$, then the rank distribution of $C$ satisfies
\begin{eqnarray*}
A_{n-d^\bot+r}(\mC) &=& (-1)^{r}q^{r \choose 2}
\sum\limits_{j=d^\bot}^{n-d}{j\brack d^\bot-r}  {j-d^\bot+r-1\brack
r-1}A_{n-j}(\mC)\\
&& + {n\brack d^\bot-r}\sum\limits_{i=0}^{r-1} (-1)^iq^{i \choose 2
}{n-d^\bot+r\brack i}    \left( q^{m(k-d^\perp+r-i)}  -1 \right).
\end{eqnarray*}
In particular, $n$, $m$, $d$, $d^\perp$, and $A_d(C),\ldots, A_{n-d^\bot}(C)$
determine the rank
distribution of $C$.
\end{corollary}

\begin{corollary} \label{dually_gabid}
Let $C\subseteq \F_{q^m}^n$ be a Gabidulin code with minimum distance $d$ and dual minimum distance $d^\perp$.
Assume that $d+d^\perp=n$. Then $A_d(C)=A_{d^\perp}(C^\perp)$.
\end{corollary}

\begin{remark}
The rank defect of 
$C$ is $\defect(C)=n+1-k-d$.
The following are easy consequences of
Proposition \ref{ssbound}.
\begin{enumerate}
 \item $\defect(C)=0$ if and only if  $C$ is MRD.
\item $\defect(C)=\defect(C^\perp)=1$ if and only if $d+d^\perp=n$. 
\end{enumerate}
Following the standard terminology, we say that $C$ is \textbf{AMRD} (Almost MRD)
if $\defect(C)=1$. We say that $C$ is \textbf{dually AMRD} if
$\defect(C)=\defect(C^\perp)=1$. Then Corollary \ref{def_gabid} 
is the analogue of \cite[Theorem 9]{Faldum-Willems}, and Corollary \ref{dually_gabid} 
is the analogue of \cite[Proposition 14]{Faldum-Willems}. 
Notice that the proof of Corollary \ref{dually_gabid} is substantially different from 
the proof of 
\cite[Proposition 14]{Faldum-Willems}.
\end{remark}

\end{document}